\documentclass[conference,10pt]{IEEEtran}
\usepackage{cite,color,graphicx,url,amssymb,amsthm,threeparttable,multirow,algorithm,algorithmic}
\usepackage[tight,footnotesize]{subfigure}
\usepackage[cmex10]{amsmath}
\IEEEoverridecommandlockouts

\newcommand{\bone}{\mathbf{1}}

\newcommand{\cA}{\mathcal{A}}

\newcommand{\C}{\mathbb{C}}
\newcommand{\cC}{\mathcal{C}}

\newcommand{\E}{\mathbb{E}}

\newcommand{\cG}{\mathcal{G}}

\newcommand{\cS}{\mathcal{S}}
\newcommand{\whcS}{\widehat{\cS}}

\newcommand{\teta}{\tilde{\eta}}

\newcommand{\tH}{\mathrm{H}}
\newcommand{\tT}{\mathrm{T}}

\newcommand{\SNR}{\textsf{\textsc{snr}}}
\newcommand{\MAR}{\textsf{\textsc{mar}}}
\newcommand{\CN}{\textsf{CN}}

\theoremstyle{plain}
\newtheorem{lemma}{Lemma}
\newtheorem{theorem}{Theorem}

\theoremstyle{definition}
\newtheorem{definition}{Definition}
\theoremstyle{remark}
\newtheorem{remark}{Remark}

\begin{document}

\title{Model Selection: Two Fundamental Measures of Coherence and Their Algorithmic Significance}

\author{\IEEEauthorblockN{Waheed U. Bajwa, Robert Calderbank, and Sina Jafarpour}
\IEEEauthorblockA{Princeton University, Princeton, NJ 08544}
}

\maketitle

\begin{abstract}
The problem of model selection arises in a number of contexts, such as compressed sensing, subset selection in linear regression, estimation of structures in graphical models, and signal denoising. This paper generalizes the notion of \emph{incoherence} in the existing literature on model selection and introduces two fundamental measures of coherence---termed as the worst-case coherence and the average coherence---among the columns of a design matrix. In particular, it utilizes these two measures of coherence to provide an in-depth analysis of a simple one-step thresholding (OST) algorithm for model selection. One of the key insights offered by the ensuing analysis is that OST is feasible for model selection as long as the design matrix obeys an easily verifiable property. In addition, the paper also characterizes the model-selection performance of OST in terms of the worst-case coherence, $\mu$, and establishes that OST performs near-optimally in the low signal-to-noise ratio regime for $N \times \cC$ design matrices with $\mu \approx O(N^{-1/2})$. Finally, in contrast to some of the existing literature on model selection, the analysis in the paper is nonasymptotic in nature, it does not require knowledge of the true model order, it is applicable to generic (random or deterministic) design matrices, and it neither requires submatrices of the design matrix to have full rank, nor does it assume a statistical prior on the values of the nonzero entries of the data vector.
\end{abstract}


\section{Introduction}
In information processing problems involving high-dimensional data, the ``curse of dimensionality'' can often be broken by exploiting the fact that real-world data tend to live in low-dimensional manifolds. This phenomenon is exemplified by the important special case in which a data vector $\alpha \in \C^\cC$ satisfies $\|\alpha\|_0 \doteq \sum_{i=1}^\cC 1_{\{|\alpha_i| > 0\}} \leq k \ll \cC$ and is observed according to the linear measurement model $f = \Phi\alpha + \eta$. Here, $\Phi$ is an $N \times \cC$ (real- or complex-valued) matrix called the \emph{measurement} or \emph{design matrix}, while $\eta \in \C^N$ represents noise in the measurement system. In this problem, the fact that $\alpha$ is ``$k$-sparse'' allows one to operate in the so-called ``compressed'' setting, $k < N \ll \cC$, thereby enabling tasks that might be deemed prohibitive otherwise.

Fundamentally, given a measurement vector $f = \Phi\alpha + \eta$ in the compressed setting, there are two complementary---but nonetheless distinct---questions that one needs to answer:
\begin{itemize}
 \item[] \hspace{-1em}\textbf{[Estimation]} Under what conditions can a $k$-sparse $\alpha$ be reliably and efficiently reconstructed from $f$?
 \item[] \hspace{-1em}\textbf{[Model Selection]} Under what conditions can the locations of the nonzero entries of a $k$-sparse $\alpha$ be reliably and efficiently recovered from $f$?
\end{itemize}
A number of researchers have successfully addressed the estimation question over the past few years under the rubric of \emph{compressed sensing}. In many application areas, however, the model-selection question is equally---if not more---important than the estimation question. In particular, the problem of model selection (sometimes also known as \emph{variable selection} or \emph{sparsity pattern recovery}) arises indirectly in a number of contexts, such as subset selection in linear regression \cite{miller:90}, estimation of structures in graphical models \cite{meinshausen:annstat06}, and signal denoising \cite{donoho:siamjsc98}. In addition, solving the model-selection problem sometimes also enables one to solve the estimation problem.%

In this paper, we study the problem of \emph{polynomial-time model selection in a compressed setting for the case when the true model order $k$ is unknown}. Despite being well-motivated by applications, this problem has received less attention compared to its estimation counterpart in the compressed sensing literature; the most notable exceptions here being \cite{meinshausen:annstat06, zhao:jmlr06, wainwright:tit09, candes:annstat09, fletcher:tit09, reeves:asilomar09}. In particular, the results reported in \cite{meinshausen:annstat06,zhao:jmlr06} establish that the lasso \cite{tibshirani:jrss96} asymptotically identifies the correct model under certain conditions on the design matrix $\Phi$ and the sparse vector $\alpha$. Later, Wainwright in \cite{wainwright:tit09} strengthens the results of \cite{meinshausen:annstat06,zhao:jmlr06} and makes explicit the dependence of model selection using the lasso on the smallest (in magnitude) nonzero entry of $\alpha$. However, apart from the fact that the results reported in \cite{meinshausen:annstat06,zhao:jmlr06,wainwright:tit09} are asymptotic in nature, the main limitation of these works is that explicit verification of the conditions that $\Phi$ needs to satisfy is computationally intractable for $k \gtrsim \sqrt{N}$.

The most general (and nonasymptotic) results for model selection using the lasso have been reported in \cite{candes:annstat09}. Specifically, Cand\`{e}s and Plan establish in \cite{candes:annstat09} that the lasso correctly identifies \emph{most} models with probability $1 - O(\cC^{-1})$ under certain conditions on the smallest nonzero entry of $\alpha$ provided: (i) the spectral norm (the largest singular value) and the worst-case coherence (the maximum absolute innerproduct between the columns) of $\Phi$ are not too large, and (ii) the values of the nonzero entries of $\alpha$ are statistically independent (and statistically symmetric around zero). The main limitation of this work is that statistical independence among the nonzero entries of $\alpha$ can be difficult to ensure in many applications.%

Finally, as opposed to the approach taken in \cite{meinshausen:annstat06, zhao:jmlr06, wainwright:tit09, candes:annstat09}, the focus in \cite{fletcher:tit09,reeves:asilomar09} is on model selection using a simple thresholding algorithm. In particular, it is shown in both \cite{fletcher:tit09,reeves:asilomar09} that model selection using thresholding is asymptotically optimal in the low signal-to-noise ratio (\SNR) regime. However, one of the main limitations of these works is that the reported results are mainly asymptotic in nature and rely on having some knowledge of the true model order. In addition, the analysis carried out in \cite{fletcher:tit09} is for the specific case of an independent and identically distributed (i.i.d.) Gaussian design matrix, while the analysis carried out in \cite{reeves:asilomar09} is for the specific case of $\alpha$ with i.i.d. Gaussian nonzero entries.

\subsection{Our Contributions}
We begin by assuming that the design matrix $\Phi$ has unit $\ell_2$-norm columns and introducing two fundamental measures of coherence among the columns $\{\varphi_i \in \C^N\}$ of $\Phi$:
\begin{itemize}
\item \textbf{Worst-Case Coherence}: $\mu \doteq \max\limits_{i,j:i \neq j} \big|\langle\varphi_i, \varphi_j\rangle\big|$, and 
\item \textbf{Average Coherence}: $\nu \doteq \frac{1}{\mathcal{C}-1} \max\limits_{i} \bigg|\!\sum\limits_{j:j \neq i} \langle\varphi_i, \varphi_j\rangle\bigg|.$
\end{itemize}
In words, worst-case coherence is a similarity measure between the columns of a design matrix and average coherence is a measure of the spread of the columns of a design matrix within the $N$-dimensional unit ball. Our main objective in this paper is to make use of these two measures of coherence in order to analyze the \emph{one-step thresholding} (OST) algorithm (see Algorithm~\ref{alg:OSTA}) for model selection. Algorithmically, this makes our approach to model selection somewhat similar to the one studied by Fletcher, Rangan, and Goyal \cite{fletcher:tit09} and Reeves and Gastpar \cite{reeves:asilomar09}. Analytically, however, the results reported in this paper are more general in nature than the ones in \cite{fletcher:tit09, reeves:asilomar09}; in particular, the asymptotic results of \cite{fletcher:tit09, reeves:asilomar09} for thresholding can be obtained as a special case of Theorem~\ref{thm:measurements} in Section~\ref{sec:mainres}.

More specifically, Theorem~\ref{thm:measurements} holds for any (random or deterministic) design matrix with sufficiently small values of the worst-case and average coherence, and the stated result in that case is completely nonasymptotic in nature. Equally importantly, unlike the case of \cite{fletcher:tit09,reeves:asilomar09}, the threshold value in Theorem~\ref{thm:measurements} is completely independent of the model order $k$ and relies only on the knowledge of $\mu, \cC$, and $\SNR$. In addition, Theorem~\ref{thm:measurements} can also be combined with the necessary conditions for asymptotically consistent model selection reported in \cite{fletcher:tit09,wainwright:tit08sub} to conclude that model selection using the OST is asymptotically optimal in the low $\SNR$ regime for any design matrix that has $\mu \approx O(N^{-1/2})$ and $\nu \approx O(N^{-1})$.

Finally, in order to compare the results obtained in this paper for model selection using the OST with the nonasymptotic results reported in \cite[Theorem~1.3]{candes:annstat09} for the lasso, Theorem~\ref{thm:sigclass} rederives Theorem~\ref{thm:measurements} in terms of conditions on the model order $k$ and the smallest nonzero entry of $\alpha$. In particular, it can be easily concluded from Theorem~\ref{thm:sigclass} \mbox{and \cite[Theorem~1.3]{candes:annstat09}} that the OST---despite being computationally primitive---performs as well as the lasso for model selection in the low $\SNR$ regime provided the design matrix has $\mu \approx O(N^{-1/2})$ and $\nu \approx O(N^{-1})$. In addition, unlike the assumptions made in \cite{candes:annstat09}, the OST achieves this without requiring that most $N \times k$ submatrices of $\Phi$ be well-conditioned and the nonzero entries of $\alpha$ be statistically independent.

\section{Main Result}\label{sec:mainres}
\subsection{Problem Setup}
Before proceeding with presenting the main result of this paper, we need to be precise about our problem formulation. To this end, we begin by reconsidering the measurement model $f = \Phi\alpha + \eta$ in the compressed setting ($k < N \ll \cC$) and take the noise vector $\eta$ to be distributed as $\CN(0,\sigma^2 I)$, although the results can be readily generalized for other noise distributions. We also assume without loss of generality that $\Phi$ has unit $\ell_2$-norm columns and $\|\alpha\|_2 = 1$, since any scaling of $\Phi$ and $\alpha$ can be accounted for in the scaling of $\sigma$. In addition, we do not impose any prior distribution on the design matrix $\Phi$ and the nonzero entries of $\alpha$. Finally, we use the notation $supp(\alpha)$ for the set containing the locations of the nonzero entries of $\alpha$ and assume, similar to the case of \cite{candes:annstat09,reeves:asilomar09}, that $supp(\alpha)$ is a uniformly random $k$-subset of $\{1,\dots,\cC\}$. In other words, we have a uniform prior on the model $supp(\alpha)$.

\subsection{The Coherence Property and Its Implications}
\begin{algorithm}[t]
\caption{The One-Step Thresholding (OST) Algorithm for Model Selection}
\label{alg:OSTA}
\textbf{Input:} An $N \times \cC$ (real- or complex-valued) matrix $\Phi$, a vector $f \in \C^N$, and a thresholding parameter $\lambda > 0$.\\
\textbf{Output:} Compute $y \doteq \Phi^\tH f$ and return an estimate of the model $\whcS \doteq \left\{i \in \{1,\dots,\cC\} : |y_i| > \lambda\right\}$.
\end{algorithm}
It is often realized in the literature that successful model selection requires the columns of the design matrix to be \emph{incoherent}; see, e.g., \cite{meinshausen:annstat06,zhao:jmlr06,candes:annstat09}. Below, we mathematically formalize this notion in terms of the coherence property.
\begin{definition}[The Coherence Property]
A matrix $\Phi$ is said to obey the coherence property if the following conditions hold:
\begin{align}
\nonumber
&(\textbf{CP-1}) \quad 
\mu \leq \frac{1}{\sqrt{10\log{\cC}}}\,, \quad \text{and} \quad
&(\textbf{CP-2}) \quad
\nu \leq \frac{12\mu}{\sqrt{N}}\,.
\end{align}
\end{definition}
\noindent Notice that the coherence property can be easily verified in polynomial time since it only requires checking that $\|\Phi^\tH\Phi - I\|_{\max} \leq (10\log{\cC})^{-1/2}$ and $\|(\Phi^\tH\Phi - I)\bone\|_\infty \leq 12(\cC-1)N^{-1/2}\|\Phi^\tH\Phi - I\|_{\max}$.

Before proceeding with describing the implications of the coherence property, it is instructive to first define three fundamental quantities as follows:
\begin{align*}
\alpha_{\min} \doteq \min_{i : |\alpha_i| > 0} |\alpha_i|, \ \SNR_{\min} \doteq \frac{\alpha_{\min}^2}{\E[\|\eta\|_2^2]/k},\ \MAR \doteq \frac{\alpha_{\min}^2}{1/k}.
\end{align*}
In words, $\alpha_{\min}$ is the magnitude of the smallest nonzero entry of $\alpha$, $\SNR_{\min}$ is the ratio of the energy in the smallest nonzero entry of $\alpha$ and the average noise energy per nonzero entry, and $\MAR$---which is termed as the \emph{minimum-to-average ratio} \cite{fletcher:tit09}---is the ratio of the energy in the smallest nonzero entry of $\alpha$ and the average signal energy per nonzero entry. We are now ready to state the main result of this paper.
\begin{theorem}\label{thm:measurements}
Suppose that $\Phi$ obeys the coherence property and write its worst-case coherence as $\mu = c_1 N^{-1/\beta}$ for some \mbox{$c_1 > 0$} (which may depend on $\log{\cC}$) and $\beta \in (1,\infty]$. Next, choose the threshold $\lambda = 4\,\max\big\{12\mu\sqrt{2\log{\cC}},\sqrt{\sigma^2\log \cC}\,\big\}$. Then the OST satisfies $\Pr(\whcS \not= supp(\alpha)) \leq 9\cC^{-1}$ as long as the number of measurements
\begin{align*}
N > \max\left\{2k\log{\cC}, \frac{64}{\SNR_{\min}} k\log{\cC}, \left(\frac{2c_2}{\MAR} k\log{\cC}\right)^{\beta/2}\right\}.
\end{align*}
Here, the quantity $c_2 > 0$ is defined as $c_2 = (96\,c_1)^2$.
\end{theorem}
\begin{remark}
The constants in the second and third terms in the $\max$ expression can be significantly reduced if one is only interested in showing the model-selection consistency of OST; that is, $\lim_{\cC \rightarrow \infty}\Pr(\whcS \not= supp(\alpha)) = 0$. One should be particularly vigilant of this fact while comparing these results to the asymptotic ones reported in \cite{fletcher:tit09} for thresholding.
\end{remark}

Note that there are two fundamental but complementary approaches that can be taken while analyzing an algorithm for model selection, namely, the \emph{minimum measurement resources} approach and the \emph{permissible signal class} approach. The statement of Theorem~\ref{thm:measurements} helps us analyze the OST for model selection using the former approach and is best suited for comparing our results with those in \cite{wainwright:tit09,fletcher:tit09,wainwright:tit08sub}. On the other hand, Cand\`{e}s and Plan in \cite{candes:annstat09} take the latter approach while analyzing the lasso for model selection and the following result is best suited for comparison purposes in this regard.
\begin{theorem}\label{thm:sigclass}
Suppose that $\Phi$ obeys the coherence property and choose the threshold $\lambda = 4\,\max\big\{12\mu\sqrt{2\log{\cC}},\sqrt{\sigma^2\log \cC}\,\big\}$. Then, as long as $k \leq N/(2\log{\cC})$ and
\begin{align*}
\alpha_{\min} > \max\left\{8\sqrt{\sigma^2\log{\cC}}, 96\mu\sqrt{2\log{\cC}}\right\}
\end{align*}
the OST satisfies $\Pr(\whcS \not= supp(\alpha)) \leq 9\cC^{-1}$.
\end{theorem}

\subsection{Discussion}
The statements of Theorem~\ref{thm:measurements} and Theorem~\ref{thm:sigclass} can be best put into perspective by considering specific examples of design matrices. Because of space constraints, we only consider here the case when $\Phi$ is an (appropriately normalized) i.i.d. Gaussian matrix. It is a well-known fact in the literature that the worst-case coherence of $\Phi$ in this case is roughly $\mu \approx \sqrt{2\log{\cC}/N}$ with high probability; see, e.g., \cite{candes:annstat09}. In addition, it can also be shown using the Bernstein inequality that $\nu \leq 12N^{-1}\sqrt{2\log{\cC}}$ with high probability in this case. It therefore follows that a Gaussian design matrix obeys the coherence property with high probability.

Theorem~\ref{thm:measurements} therefore implies that the OST identifies the correct model in this case with high probability as long as \mbox{$N \gtrsim \max\left\{ \frac{k\log{\cC}}{\SNR_{\min}}, \frac{k\log^2{\cC}}{\MAR}\right\}$}. In particular, this expression reduces to $N \gtrsim \frac{k\log{\cC}}{\SNR_{\min}}$ for the case of $\SNR \doteq 1/\E[\|\eta\|_2^2] \ll 1$. On the other hand, we have from \cite{fletcher:tit09,wainwright:tit08sub} that no scheme can asymptotically identify the correct model if $N \not\gtrsim \frac{k\log{\cC}}{\SNR_{\min}}$. This proves the near-optimality of the OST for model selection in the low $\SNR$ regime for any design matrix that has $\mu \approx O(N^{-1/2})$ and $\nu \approx O(N^{-1})$. Finally, note that we could have made a similar conclusion by focusing on Theorem~\ref{thm:sigclass} and comparing the conditions in the low $\SNR$ regime in that case with those in \cite[Theorem~1.3]{candes:annstat09}.

\section{Proofs}
The general roadmap for the proofs of Theorem~\ref{thm:measurements} and Theorem~\ref{thm:sigclass} is as follows. Below, we first introduce the notion of $(k,\epsilon,\delta)$-\emph{statistical orthogonality condition} (StOC). Next, we establish in Lemma~\ref{lem:osta} that if $\Phi$ satisfies the StOC then OST recovers the support of $\alpha$ with high probability provided $\alpha_{\min}$ is large enough. Subsequently, we establish in Lemma~\ref{lem:soc1} and Lemma~\ref{lem:soc2} the relationship between the StOC parameters and the worst-case and average coherence of $\Phi$. The proofs of Theorem~\ref{thm:measurements} and Theorem~\ref{thm:sigclass} then follow by judiciously combining the results of these three lemmas using the coherence property.
\begin{definition}[$(k,\epsilon,\delta)$-StOC]
Let $\Pi \doteq (\pi_1,\dots,\pi_k)$ be a uniformly random (ordered) $k$-subset of $\{1,\dots,\cC\}$ and let $\Pi^c \doteq \{1,\dots,\cC\} - \Pi$. Then, given $\epsilon, \delta \in [0,1)$, $\Phi$ is said to satisfy the $(k,\epsilon,\delta)$-statistical orthogonality condition if the following inequalities
\begin{align}
\nonumber
&(\textbf{StOC-1}) 
&\big\|(\Phi_\Pi^\tH \Phi_\Pi - I)z\big\|_\infty \leq \epsilon \|z\|_2\\
\nonumber
&(\textbf{StOC-2})
&\big\|\Phi^\tH_{\Pi^c} \Phi_\Pi z\big\|_\infty \leq \epsilon \|z\|_2
\end{align}
hold for every \emph{fixed} $z \in \C^k$ with probability exceeding $1 - \delta$ (with respect to the choice of $\Pi$).
\end{definition}
\begin{remark}
Note that the StOC derives its name from the fact that if $\Phi$ is an orthogonal matrix then it trivially satisfies the StOC for every $k$ with $\epsilon = \delta = 0$.
\end{remark}
\begin{lemma}\label{lem:osta}
Let $\Pi \doteq supp(\alpha)$ be a uniformly random \mbox{$k$-subset} of $\{1,\dots,\cC\}$. Further, suppose that the matrix $\Phi$ satisfies the $(k,\epsilon,\delta)$-StOC and choose the threshold as $\lambda = 2\,\max\big\{\epsilon,2\sqrt{\sigma^2\log \cC}\,\big\}$. Then, under the assumption that $\alpha_{\min} > 2\lambda$, the OST satisfies
\begin{align}
\nonumber
	\Pr\left(\whcS \not= \Pi\right) \leq \delta + 2\left(\sqrt{2\pi \log{\cC}}\cdot \cC\right)^{-1}.
\end{align}
\end{lemma}
\begin{proof}
We begin by defining $z^\tT \doteq \begin{bmatrix}\alpha_{\pi_1} & \dots & \alpha_{\pi_k}\end{bmatrix}$ and writing the vector $y = \Phi^\tH f$ as $y = \Phi^\tH\Phi_\Pi z + \Phi^\tH\eta$. Now, let $\Pi^c \doteq \{1,\dots,\cC\} - \Pi$ and note that in order to establish that $\whcS = \Pi$ we need to show that $\|y_{\Pi^c}\|_\infty \leq \lambda$ and $\min\limits_{i}|y_{\pi_i}| > \lambda$.

In this regard, first note that $\teta \doteq \Phi^\tH\eta$ is a complex Gaussian random vector whose entries are identically (although not independently) distributed as $\CN(0,\sigma^2)$. It therefore follows from the tail bound on the maximum of $\cC$ \emph{arbitrary} complex Gaussian random variables that $\|\teta\|_\infty \leq 2\sqrt{\sigma^2\log \cC}$ with probability exceeding $1 - 2\left(\sqrt{2\pi \log{\cC}}\cdot \cC\right)^{-1}$. Further, define $$\cG \doteq \big\{\{\|\teta\|_\infty \leq 2\sqrt{\sigma^2\log \cC}\} \bigcap \{(\textbf{StOC-1}) \cap (\textbf{StOC-2})\}\big\}$$ and notice that, since the noise is independent of $\Pi$, we have $\Pr(\cG) > 1 - \delta - 2\left(\sqrt{2\pi \log{\cC}}\cdot \cC\right)^{-1}$. In addition, conditioned on the event $\cG$, we have
\begin{align}
\nonumber
	\|y_{\Pi^c}\|_\infty &\stackrel{(a)}{\leq} \|\Phi_{\Pi^c}^\tH\Phi_\Pi z\|_\infty + \|\teta\|_\infty\\
	&\stackrel{(b)}{\leq} \epsilon + 2\sqrt{\sigma^2\log \cC} \stackrel{(c)}{\leq} \lambda 
\end{align}
where $(a)$ follows from the triangle inequality, $(b)$ is a consequence of the conditioning on $\cG$, and $(c)$ follows from the fact that $\lambda = 2\,\max\big\{\epsilon,2\sqrt{\sigma^2\log \cC}\,\big\}$.

Finally, in order to show that $\min\limits_{i}|y_{\pi_i}| > \lambda$, we define $r = (\Phi_\Pi^\tH \Phi_\Pi - I)z$ and note that, conditioned on the event $\cG$, we have for any $i \in \{1,\dots,k\}$
\begin{align}
\nonumber
 |y_{\pi_i}| = |\alpha_{\pi_i} + r_i + \teta_{\pi_i}| &\geq |\alpha_{\pi_i}| - \|r\|_\infty - \|\teta\|_\infty\\
 &\stackrel{(d)}{>} 2\lambda - \epsilon - 2\sqrt{\sigma^2\log \cC} \stackrel{(e)}{\geq} \lambda
\end{align}
where $(d)$ follows from the conditioning on $\cG$ and the assumption that $\alpha_{\min} > 2\lambda$, while $(e)$ is a simple consequence of the choice of $\lambda$. This completes the proof of the lemma since we have now shown that $\Pr(\whcS \not= \Pi) \leq \Pr\left(\cG^c\right)$.
\end{proof}

Having established Lemma~\ref{lem:osta}, our next goal is to relate the StOC parameters with the worst-case and average coherence.
\begin{lemma}\label{lem:soc1}
Let $\Pi \doteq (\pi_1,\dots,\pi_k)$ be a uniformly random (ordered) $k$-subset of $\{1,\dots,\cC\}$. Then, for any fixed $z \in \C^k$, $\epsilon \geq 0$, and $k \leq \min\{\epsilon^2\nu^{-2}/4,\cC/2\}$, we have
\begin{align}
\nonumber
\Pr\left(\{\Phi \text{ \emph{does not satisfy}} (\textbf{\emph{StOC-1}})\}\right) \leq 4k\exp\left(-\frac{\epsilon^2 \mu^{-2}}{576}\right).
\end{align}
\end{lemma}
\begin{proof}
The proof of this lemma relies heavily on the so-called \emph{method of bounded differences} (MOBD) \cite{mcdiarmid:89}. Specifically, note that $\big\|(\Phi_\Pi^\tH \Phi_\Pi - I)z\big\|_\infty = \max\limits_i \bigg|\sum\limits_{j\not=i} z_j \langle\varphi_{\pi_i},\varphi_{\pi_j}\rangle\bigg|$ and define $\Pi^{-i} \doteq (\pi_1,\dots,\pi_{i-1},\pi_{i+1},\dots,\pi_k)$. Then for a fixed index $i$, and conditioned on the event $\cA_{i^\prime} \doteq \{\pi_i = i^\prime\}$, we have the following equality from basic probability theory
\begin{align}
\nonumber
 &\Pr\bigg(\big|\sum\limits_{\substack{j=1\\j\not=i}}^k z_j \langle\varphi_{\pi_i},\varphi_{\pi_j}\rangle\big| > \epsilon\|z\|_2\bigg|\cA_{i^\prime}\bigg) = \\
\label{eqn:lem_prob}
&\qquad\qquad\qquad \Pr\bigg(\big|\sum\limits_{\substack{j=1\\j\not=i}}^k z_j \langle\varphi_{i^\prime},\varphi_{\pi_j}\rangle\big| > \epsilon\|z\|_2\bigg|\cA_{i^\prime}\bigg).
\end{align}

Next, in order to apply the MOBD, we construct a Doob's martingale sequence $(M_0,M_1,\dots,M_{k-1})$ as follows:
\begin{align}
\nonumber
 M_0 &= \E\Big[\sum\limits_{\substack{j=1\\j\not=i}}^k z_j \langle\varphi_{i^\prime},\varphi_{\pi_j}\rangle\Big|\cA_{i^\prime}\Big], \ \text{and}\\ 
 M_\ell &= \E\Big[\sum\limits_{\substack{j=1\\j\not=i}}^k z_j \langle\varphi_{i^\prime},\varphi_{\pi_j}\rangle\Big|\pi^{-i}_{1\rightarrow\ell},\cA_{i^\prime}\Big], \ \ell=1,\dots,k-1
\end{align}
where $\pi^{-i}_{1\rightarrow\ell}$ is the first $\ell$ coordinates of $\Pi^{-i}$. Here, note that
\begin{align}
\nonumber
 \big|M_0\big| &\leq \sum_{j\not=i} \big|z_j\big| \Big|\E\big[\langle\varphi_{i^\prime},\varphi_{\pi_j}\rangle|\cA_{i^\prime}\big]\Big|\\
\label{eqn:lem_M0}
&\stackrel{(a)}{\leq} \sum_{j\not=i} \big|z_j\big| \Bigg|\sum\limits_{\substack{q=1\\q\not=i^\prime}}^\cC\frac{1}{\cC-1} \langle\varphi_{i^\prime},\varphi_{q}\rangle\Bigg| \stackrel{(b)}{\leq} \sqrt{k}\,\nu\,\|z\|_2
\end{align}
where $(a)$ follows since, conditioned on $\cA_{i^\prime}$, $\pi_j$ has a uniform distribution over $\{1,\dots,\cC\} - \{i^\prime\}$, while $(b)$ mainly follows from the definition of average coherence. Further, if we define 
\begin{align}
	M_\ell(x) \doteq \E\Big[\sum\limits_{\substack{j=1\\j\not=i}}^k z_j 	\langle\varphi_{i^\prime},\varphi_{\pi_j}\rangle\Big|\pi^{-i}_{1\rightarrow\ell-1},\pi^{-i}_\ell=x,\cA_{i^\prime}\Big]
\end{align}
then, since $(M_0,M_1,\dots,M_{k-1})$ is a Doob's martingale sequence, it can be easily verified that $|M_\ell - M_{\ell-1}|$ is upperbounded by $\sup_{x,y} \big[M_\ell(x) - M_\ell(y)\big]$ (see, e.g., \cite{motwani:95}).

Next, in order to upperbound $\sup_{x,y} \big[M_\ell(x) - M_\ell(y)\big]$, we first define $d_{\ell,j} \doteq \E\Big[\langle\varphi_{i^\prime},\varphi_{\pi_j}\rangle\Big|\pi^{-i}_{1\rightarrow\ell-1},\pi^{-i}_\ell=x,\cA_{i^\prime}\Big] - \E\Big[\langle\varphi_{i^\prime},\varphi_{\pi_j}\rangle\Big|\pi^{-i}_{1\rightarrow\ell-1},\pi^{-i}_\ell=y,\cA_{i^\prime}\Big]$ and then notice that
\begin{align}
 \Big|M_\ell(x) - M_\ell(y)\Big| \leq \sum_{\substack{j \leq\,\ell+1\\j\not=i}}\big|z_j\big|\big|d_{\ell,j}\big| + \sum_{\substack{j >\,\ell+1\\j\not=i}}\big|z_j\big|\big|d_{\ell,j}\big|.
\end{align}
In addition, we have that for every $j > \ell+1, j \not= i$, the random variable $\pi_j$ has a uniform distribution over $\{1,\dots,\cC\} - \{\pi^{-i}_{1\rightarrow\ell-1}, x, i^\prime\}$ when conditioned on $\{\pi^{-i}_{1\rightarrow\ell-1}, \pi^{-i}_\ell = x, i^\prime\}$, while $\pi_j$ has a uniform distribution over $\{1,\dots,\cC\} - \{\pi^{-i}_{1\rightarrow\ell-1}, y, i^\prime\}$ when conditioned on $\{\pi^{-i}_{1\rightarrow\ell-1}, \pi^{-i}_\ell = y, i^\prime\}$. Therefore, we have for every $j > \ell+1, j \not= i$ that
\begin{align}
	|d_{\ell,j}| = \frac{1}{\cC-\ell-1}\Big|\langle\varphi_{i^\prime},\varphi_{y}\rangle - \langle\varphi_{i^\prime},\varphi_{x}\rangle\Big| \leq \frac{2 \mu}{\cC-k}.
\end{align}
Similarly, it can be argued that $\sum_{\substack{j \leq\,\ell+1\\j\not=i}}\big|z_j\big|\big|d_{\ell,j}\big| \leq \big|z_{\ell+1}\big| 2 \mu$ when $i \leq \ell$, $\sum_{\substack{j \leq\,\ell+1\\j\not=i}}\big|z_j\big|\big|d_{\ell,j}\big| \leq \big|z_{\ell}\big| 2 \mu$ when $i = \ell+1$, and $\sum_{\substack{j \leq\,\ell+1\\j\not=i}}\big|z_j\big|\big|d_{\ell,j}\big| \leq (|z_{\ell}| + \frac{|z_{\ell+1}|}{\cC-k})2 \mu$ when $i > \ell+1$. Consequently, it can be easily verified that
\begin{align}
\label{eqn:lem_c_ell}
 \sup_{x,y} \big[M_\ell(x) - M_\ell(y)\big] \leq \underbrace{2\mu\Big(|z_\ell| + |z_{\ell+1}| + \frac{\|z\|_1}{\cC-k}\Big)}_{\doteq\,c_\ell}.
\end{align}

We have now established that $(M_0,M_1,\dots,M_{k-1})$ is a bounded-difference martingale with $|M_\ell - M_{\ell-1}| \leq c_\ell$ for $\ell=1,\dots,k-1$. Further, it can also be verified from \eqref{eqn:lem_c_ell} that $\sum_{\ell=1}^{k-1}c_\ell^2 \leq 36\mu^2\|z\|_2^2$ since $k \leq \cC/2$. In addition, since $|M_0| \leq \sqrt{k}\,\nu\,\|z\|_2$ and $k \leq \epsilon^2\nu^{-2}/4$, we have from the Azuma inequality for bounded-difference martingale sequences \cite{azuma:tmj67} adapted to the complex-valued setup that
\begin{align}
\nonumber
 &\Pr\bigg(\big|\sum\limits_{\substack{j=1\\j\not=i}}^k z_j \langle\varphi_{i^\prime},\varphi_{\pi_j}\rangle\big| > \epsilon\|z\|_2\bigg|\cA_{i^\prime}\bigg) \leq \\
\label{eqn:lem_cp1_finalbd}
 &\Pr\bigg(\big|M_{k-1} - M_0| > \frac{\epsilon\|z\|_2}{2}\bigg|\cA_{i^\prime}\bigg)
\leq 4\exp\!\big(\!\!-\frac{\epsilon^2 \mu^{-2}}{576}\!\big).
\end{align}
Combining all these facts together, we therefore finally obtain
\begin{align}
\nonumber
 &\Pr\bigg(\big\|(\Phi_\Pi^\tH \Phi_\Pi - I)z\big\|_\infty > \epsilon\|z\|_2\bigg)\\
\nonumber
	&\quad \stackrel{(c)}{\leq} k\,\sum_{i^\prime=1}^{\cC} \Pr\bigg(\big|\sum\limits_{\substack{j=1\\j\not=i}}^k z_j \langle\varphi_{i^\prime},\varphi_{\pi_j}\rangle\big| > \epsilon\|z\|_2\bigg|\cA_{i^\prime}\bigg) \Pr\left(\cA_{i^\prime}\right)\\
	&\quad\stackrel{(d)}{\leq} 4k\exp\left(-\frac{\epsilon^2 \mu^{-2}}{576}\right)
\end{align}
where $(c)$ follows from the union bound and the fact that the $\pi_i$'s are identically (although not independently) distributed, while $(d)$ follows from \eqref{eqn:lem_cp1_finalbd} and the fact that $\pi_i$ has a uniform distribution over $\{1,\dots,\cC\}$.
\end{proof}

\begin{lemma}\label{lem:soc2}
Let $\Pi \doteq (\pi_1,\dots,\pi_k)$ be a uniformly random (ordered) $k$-subset of $\{1,\dots,\cC\}$. Further, define the random subset $\Pi^c \doteq \{1,\dots,\cC\} - \Pi$. Then, for any fixed $z \in \C^k$, $\epsilon \geq 0$, and $k \leq \min\{\epsilon^2\nu^{-2}/4,\cC/2\}$, we have
\begin{align}
\nonumber
\Pr\left(\{\Phi \text{ \emph{does not satisfy}} (\textbf{\emph{StOC-2}})\}\right) \leq 4\cC\exp\left(-\frac{\epsilon^2 \mu^{-2}}{256}\right).
\end{align}
\end{lemma}
\begin{proof}[Proof Sketch]
The proof of this lemma also relies on the MOBD and is very similar to that of Lemma~\ref{lem:soc1}. As such, we only provide a sketch of the proof here. To begin with, we note that $\big\|\Phi_{\Pi^c}^\tH \Phi_\Pi z\big\|_\infty = \max\limits_{i\in[\cC-k]} \bigg|\sum\limits_{j} z_j \langle\varphi_{\pi^c_i},\varphi_{\pi_j}\rangle\bigg|$, where $[\cC-k] \doteq \{1,\dots,\cC-k\}$ and $\pi^c_i$ denotes the $i^{th}$ coordinate of $\Pi^c$. Then for a fixed index $i \in [\cC-k]$, and conditioned on the event $\cA_{i^\prime} \doteq \{\pi^c_i = i^\prime\}$, we have the following equality
\begin{align}
\nonumber
 &\Pr\bigg(\big|\sum\limits_{j=1}^k z_j \langle\varphi_{\pi^c_i},\varphi_{\pi_j}\rangle\big| > \epsilon\|z\|_2\bigg|\cA_{i^\prime}\bigg) = \\
\label{eqn:lem_soc2_prob}
&\qquad\qquad\qquad \Pr\bigg(\big|\sum\limits_{j=1}^k z_j \langle\varphi_{i^\prime},\varphi_{\pi_j}\rangle\big| > \epsilon\|z\|_2\bigg|\cA_{i^\prime}\bigg).
\end{align}

Next, as in the case of Lemma~\ref{lem:soc1}, we construct a Doob's martingale sequence $(M_0,M_1,\dots,M_k)$ as follows:
\begin{align}
\nonumber
 M_0 &= \E\Big[\sum\limits_{j=1}^k z_j \langle\varphi_{i^\prime},\varphi_{\pi_j}\rangle\Big|\cA_{i^\prime}\Big], \ \text{and}\\ 
 M_\ell &= \E\Big[\sum\limits_{j=1}^k z_j \langle\varphi_{i^\prime},\varphi_{\pi_j}\rangle\Big|\pi_{1\rightarrow\ell},\cA_{i^\prime}\Big], \ \ell=1,\dots,k
\end{align}
where $\pi_{1\rightarrow\ell}$ now denotes the first $\ell$ coordinates of $\Pi$. It can now be argued (as in Lemma~\ref{lem:soc1}) that: (i) $\big|M_0\big| \leq \sqrt{k}\,\nu\,\|z\|_2$, (ii) $|M_\ell - M_{\ell-1}| \leq 2\mu\Big(|z_\ell| + \frac{\|z\|_1}{\cC-k}\Big) \doteq c_\ell$, and (iii) $\sum_{\ell=1}^{k}c_\ell^2 \leq 16\mu^2\|z\|_2^2$. Therefore, since $k \leq \epsilon^2\nu^{-2}/4$, we once again have from the (complex) Azuma inequality that \eqref{eqn:lem_soc2_prob} is upperbounded by $4\exp\big(-\frac{\epsilon^2 \mu^{-2}}{256}\big)$. Combining all these facts together, we finally obtain the claimed result as follows
\begin{align}
\Pr\bigg(\big\|\Phi_{\Pi^c}^\tH \Phi_\Pi z\big\|_\infty > \epsilon\|z\|_2\bigg) \stackrel{(a)}{\leq} 4\cC\exp\left(-\frac{\epsilon^2 \mu^{-2}}{256}\right)
\end{align}
where $(a)$ mainly follows from the union bound and the fact that $\pi^c_i$ has a uniform distribution over $\{1,\dots,\cC\}$.
\end{proof}

We are finally ready to prove the main results of this paper using Lemmata~\ref{lem:osta}--\ref{lem:soc2}.
\begin{proof}[Proof of Theorem~\ref{thm:measurements}]
Note that Lemma~\ref{lem:soc1} and Lemma~\ref{lem:soc2} imply that if $\Phi$ has worst-case coherence $\mu$ and average coherence $\nu$ then, as long as $k \leq \cC/2$, $\Phi$ satisfies the $(k,\epsilon,\delta)$-StOC for any $\epsilon \in [2\sqrt{k}\nu, 1)$ with $\delta = 8\cC\exp\left(-\frac{\epsilon^2 \mu^{-2}}{576}\right)$.

Now let $k \leq N/(2\log{\cC})$ and define $\epsilon^\prime \doteq 24\mu\sqrt{2\log{\cC}}$. Then, since $\Phi$ satisfies the coherence property, we have $2\sqrt{k}\nu \leq \epsilon^\prime < 1$ and therefore $\Phi$ satisfies the $(k,\epsilon^\prime,\delta^\prime)$-StOC with $\delta^\prime \doteq 8\cC^{-1}$. Consequently, Lemma~\ref{lem:osta} states that $\Pr(\whcS \not= supp(\alpha)) \leq 9\cC^{-1}$ as long as $N \geq 2k\log{\cC}$, $\alpha_{\min} > 4\epsilon^\prime$, and $\alpha_{\min} > 8\sqrt{\sigma^2\log{\cC}}$. Further, note that
\begin{align*}
\alpha_{\min} > 8\sqrt{\sigma^2\log{\cC}} \ &\Longleftrightarrow \ N > \frac{64}{\SNR_{\min}} k\log{\cC}, \ \text{and}\\
\alpha_{\min} > 4\epsilon^\prime \ &\Longleftrightarrow \ N > \left(\frac{2c_2}{\MAR} k\log{\cC}\right)^{\beta/2}.
\end{align*}
This completes the proof of the theorem.
\end{proof}
\begin{proof}[Proof of Theorem~\ref{thm:sigclass}]
The proof of this theorem follows along similar lines as that of Theorem~\ref{thm:measurements} and is therefore omitted here for the sake of brevity.
\end{proof}

\section{Conclusion}
In this paper, we have analyzed the one-step thresholding (OST) algorithm for model selection in terms of the worst-case and average coherence of the design matrix. In stark contrast to the existing work on model selection using thresholding, our analysis is completely nonasymptotic in nature, it does not require knowledge of the true model order, and it is applicable to arbitrary (random or deterministic) design matrices. In particular, we have established in the paper that the OST can be used for model selection as long as the design matrix obeys an easily verifiable property. Further, we have specified the dependence of the OST performance on the worst-case coherence of the design matrix and shown that it performs near-optimally in the low $\SNR$ regime for design matrices with $O(N^{-1/2})$ worst-case coherence. Finally, unlike the assumptions made in \cite{candes:annstat09}, our analysis also does not require that most $N \times k$ submatrices of $\Phi$ be well-conditioned and the nonzero entries of the data vector be statistically independent.


\end{document}